\documentclass[1p,final]{elsarticle}
\usepackage{amsfonts,color,morefloats,pslatex}
\usepackage{amssymb,amsthm, amsmath,latexsym}

\newtheorem{theorem}{Theorem}

\newtheorem{example}[theorem]{Example}

\newcommand{\lcm}{{\mathrm{lcm}}}

\newcommand{\gf}{{\mathrm{GF}}}

\newcommand{\m}{\mathbb{M}}

\newcommand{\cP}{{\mathcal{P}}}
\newcommand{\cB}{{\mathcal{B}}}

\newcommand{\C}{{\mathcal{C}}}

\newcommand{\bzero}{{\mathbf{\bar{0}}}}
\newcommand{\bone}{{\mathbf{\bar{1}}}}

\newcommand{\bD}{{\mathbb{D}}}

\begin{document}

\begin{frontmatter}



\title{Infinite families of $t$-designs from a type of five-weight codes}
%\tnotetext[fn1]{C. Ding's research was supported by the Hong Kong Research Grants Council,
%Proj. No. 123456.}

\author{Cunsheng Ding}
\ead{cding@ust.hk}
%\author[lcj]{Chengju Li}
%\ead{lichengju1987@163.com}

%\cortext[lcj]{Corresponding author}
\address{Department of Computer Science
                                                  and Engineering, The Hong Kong University of Science and Technology,
                                                  Clear Water Bay, Kowloon, Hong Kong, China}
%\address[lcj]{School of Computer Science and Software Engineering, East China Normal University, Shanghai, 200062, China}

\begin{abstract}
It has been known for a long time that $t$-designs can be employed to construct both linear and nonlinear codes 
and that the codewords of a fixed weight in a code may hold a $t$-design. While a lot of progress in the direction 
of constructing codes from $t$-designs has been made, only a small amount of work on the construction of 
$t$-designs from codes has been done. The objective of this paper is to construct infinite families of $2$-designs 
and $3$-designs from a type of binary linear codes with five-weights. The total number of $2$-designs and $3$-designs 
obtained in this paper are exponential in any odd $m$ and the block size of the designs varies in a huge range. 
\end{abstract}


\end{frontmatter}

\section{Introduction}

We start with a brief recall of $t$-designs. 
Let $\cP$ be a set of $v \ge 1$ elements, and let $\cB$ be a set of $k$-subsets of $\cP$, where $k$ is
a positive integer with $1 \leq k \leq v$. Let $t$ be a positive integer with $t \leq k$. The pair
$\bD = (\cP, \cB)$ is called a $t$-$(v, k, \lambda)$ {\em design\index{design}}, or simply {\em $t$-design\index{$t$-design}}, if every $t$-subset of $\cP$ is contained in exactly $\lambda$ elements of
$\cB$. The elements of $\cP$ are called points, and those of $\cB$ are referred to as blocks.
We usually use $b$ to denote the number of blocks in $\cB$.  A $t$-design is called {\em simple\index{simple}} if $\cB$ does not contain repeated blocks. In this paper, we consider only simple 
$t$-designs.  A $t$-design is called {\em symmetric\index{symmetric design}} if $v = b$. It is clear that $t$-designs with $k = t$ or $k = v$ always exist. Such $t$-designs are {\em trivial}. In this paper, we consider only $t$-designs with $v > k > t$.
A $t$-$(v,k,\lambda)$ design is referred to as a {\em Steiner system\index{Steiner system}} if $t \geq 2$ and $\lambda=1$, and is denoted by $S(t,k, v)$.

A necessary condition for the existence of a $t$-$(v, k, \lambda)$ design is that
\begin{eqnarray}\label{eqn-tdesignnecessty}
\binom{k-i}{t-i} \mbox{ divides } \lambda \binom{v-i}{t-i}
\end{eqnarray}
for all integer $i$ with $0 \leq i \leq t$.

The interplay between codes and $t$-designs goes in two directions. In one direction, the incidence matrix 
of any $t$-design generates a linear code over any finite field $\gf(q)$. A lot of progress in this direction has 
been made and documented in the literature (see, for examples, \cite{AK92}, \cite{DingBook}, 
\cite{Tonchev,Tonchevhb}). In the other direction, 
the codewords of a fixed Hamming weight in a linear or nonlinear code may hold a $t$-design. 
Some linear and nonlinear codes were employed to construct $t$-designs  
\cite{AK92,HKM04,HMT05,KM00,MT04,Tonchev,Tonchevhb}. Binary and ternary Golay codes of certain parameters give  
$4$-designs and $5$-designs. However, the largest $t$ for which an infinite family of $t$-designs 
is derived directly from codes is $t=3$. According to the references \cite{AK92}, 
\cite{Tonchev,Tonchevhb} and \cite{KLhb}, not much progress on the construction of $t$-designs from codes has 
been made so far, while many other constructions of $t$-designs are documented in the 
literature (\cite{BJL,CMhb,KLhb,RR10}).    

The objective of this paper is to construct infinite families of $2$-designs and $3$-designs from a type 
of binary linear codes with five weights. The obtained $t$-designs depend only on the weight distribution 
of the underlying binary codes. The total number of $2$-designs and $3$-designs presented in this paper 
are exponential in $m$, where $m \geq 5$ is an odd integer.  In addition, the block size of the designs 
can vary in a huge range.  

\section{The classical construction of $t$-designs from codes}

Let $\C$ be a $[v, \kappa, d]$ linear code over $\gf(q)$. Let $A_i:=A_i(\C)$, which denotes the
number of codewords with Hamming weight $i$ in $\C$, where $0 \leq i \leq v$. The sequence 
$(A_0, A_1, \cdots, A_{v})$ is
called the \textit{weight distribution} of $\C$, and $\sum_{i=0}^v A_iz^i$ is referred to as
the \textit{weight enumerator} of $\C$. For each $k$ with $A_k \neq 0$,  let $\cB_k$ denote
the set of the supports of all codewords with Hamming weight $k$ in $\C$, where the coordinates of a codeword
are indexed by $(0,1,2, \cdots, v-1)$. Let $\cP=\{0, 1, 2, \cdots, v-1\}$.  The pair $(\cP, \cB_k)$
may be a $t$-$(v, k, \lambda)$ design for some positive integer $\lambda$. The following
theorems, developed by Assumus and Mattson, show that the pair $(\cP, \cB_k)$ defined by 
a linear code is a $t$-design under certain conditions.

\begin{theorem}\label{thm-AM1}[Assmus-Mattson Theorem \cite{AM74}, \cite[p. 303]{HP03}] 
Let $\C$ be a binary $[v, \kappa, d]$ code. Suppose $\C^\perp$ has minimum weight $d^\perp$.
Suppose that $A_i=A_i(\C)$ and $A_i^\perp=A_i(\C^\perp)$, for $0 \leq i \leq v$, are the
weight distributions of $\C$ and $\C^\perp$, respectively. Fix a positive integer $t$
with $t < d$, and let $s$ be the number of $i$ with $A_i^\perp \ne 0$ for $0 < i \leq v-t$.
Suppose that $s \leq d -t$. Then
\begin{itemize}
\item the codewords of weight $i$ in $\C$ hold a $t$-design provided that $A_i \ne 0$ and
      $d \leq i \leq v$, and
\item the codewords of weight $i$ in $\C^\perp$ hold a $t$-design provided that
      $A_i^\perp \ne 0$ and $d^\perp \leq i \leq v$.
\end{itemize}
\end{theorem}

To construct $t$-designs via Theorem \ref{thm-AM1}, we will need the 
following lemma in subsequent sections, which is a variant of the MacWilliam Identity 
\cite[p. 41]{vanLint}. 

\begin{theorem} \label{thm-MI}
Let $\C$ be a $[v, \kappa, d]$ code over $\gf(q)$ with weight enumerator $A(z)=\sum_{i=0}^v A_iz^i$ and let
$A^\perp(z)$ be the weight enumerator of $\C^\perp$. Then
$$A^\perp(z)=q^{-\kappa}\Big(1+(q-1)z\Big)^vA\Big(\frac {1-z} {1+(q-1)z}\Big).$$
\end{theorem}

Later in this paper, we will need also the following theorem. 

\begin{theorem}\label{thm-allcodes}
Let $\C$ be an $[n, k, d]$ binary linear code, and let $\C^\perp$ denote the dual of $\C$. Denote by $\overline{\C^\perp}$ 
the extended code of $\C^\perp$, and let  $\overline{\C^\perp}^\perp$ denote the dual of  $\overline{\C^\perp}$. Then we 
have the following. 
\begin{enumerate}
\item $\C^\perp$ has parameters $[n, n-k, d^\perp]$, where $d^\perp$ denotes the minimum distance of $\C^\perp$. 
\item $\overline{\C^\perp}$ has parameters $[n+1, n-k, \overline{d^\perp}]$, where $\overline{d^\perp}$ denotes the minimum 
distance of $\overline{\C^\perp}$, and is given by 
\begin{eqnarray*}
\overline{d^\perp} = \left\{ 
\begin{array}{ll}
d^\perp  & \mbox{ if $d^\perp$ is even,}  \\
d^\perp  + 1 & \mbox{ if $d^\perp$ is odd.} 
\end{array} 
\right. 
\end{eqnarray*}
\item $\overline{\C^\perp}^\perp$ has parameters $[n+1, k+1, \overline{d^\perp}^\perp]$, where $\overline{d^\perp}^\perp$ denotes the minimum 
distance of $\overline{\C^\perp}^\perp$. 
Furthermore, $\overline{\C^\perp}^\perp$ has only even-weight codewords, and all the nonzero weights in $\overline{\C^\perp}^\perp$ are 
the following: 
\begin{eqnarray*}
w_1, \, w_2,\, \cdots,\, w_t;\, n+1-w_1, \, n+1-w_2, \, \cdots, \, n+1-w_t; \, n+1,  
\end{eqnarray*} 
where $w_1,\,  w_2,\, \cdots,\, w_t$ denote all the nonzero weights of $\C$. 
\end{enumerate} 
\end{theorem} 

\begin{proof}
The conclusions of the first two parts are straightforward. We prove only the conclusions of the third part below. 

Since $\overline{\C^\perp}$ has length $n+1$ and dimension $n-k$,  the dimension of $\overline{\C^\perp}^\perp$ 
is $k+1$. By assumption, all codes under consideration are binary. By definition, $\overline{\C^\perp}$ has only 
even-weight codewords. Recall that $\overline{\C^\perp}$ is the extended code of $\C^\perp$. It is known that  
the generator matrix of $\overline{\C^\perp}^\perp$ is given by  (\cite[p. 15]{HP03})  
\begin{eqnarray*}
\left[ 
\begin{array}{cc}
\bone & 1 \\
G & \bzero
\end{array}
\right]. 
\end{eqnarray*}
where $\bone=(1 1 1 \cdots 1)$ is the all-one vector of length $n$,  $\bzero=(000 \cdots 0)^T$, which is a column 
vector of length $n$, and $G$ is the generator matrix of $\C$. Notice again that  $\overline{\C^\perp}^\perp$ is binary, the 
desired conclusions on the weights in $\overline{\C^\perp}^\perp$ follow from the relation between the two generator 
matrices of the two codes  $\overline{\C^\perp}^\perp$ and $\C$. 
\end{proof}

\section{A type of binary linear codes with five-weights and related codes}\label{sec-maincodes}

In this section, we first introduce a type of binary linear codes $\C_m$ of length $n=2^m-1$, which has the weight distribution of 
Table \ref{tab-zhou3}, and then analyze their dual codes $\C_m^\perp$, the extended codes $\overline{\C_m^\perp}$, and the duals 
 $\overline{\C_m^\perp}^\perp$.  Such codes will be employed to construct $t$-designs in Sections \ref{sec-2designs} and \ref{sec-3designs}. Examples of such codes will be given 
 in Section \ref{sec-examplecodes}.

\begin{table}[ht]
\caption{The weight distribution of $\C_{m}$ for odd $m$.}\label{tab-zhou3}
\centering
\begin{tabular}{ll}
\hline
Weight $w$    & No. of codewords $A_w$  \\ \hline
$0$                                                        & $1$ \\ 
$2^{m-1}-2^{(m+1)/2}$           & $(2^m-1)\cdot 2^{(m-5)/2}\cdot (2^{(m-3)/2}+1)\cdot (2^{m-1}-1)/3$ \\ 
$2^{m-1}-2^{(m-1)/2}$           & $(2^m-1)\cdot 2^{(m-3)/2}\cdot (2^{(m-1)/2}+1)\cdot (5\cdot 2^{m-1}+4)/3$ \\ 
$2^{m-1}$           &           ${(2^m-1)}\cdot (9\cdot 2^{2m-4}+3\cdot 2^{m-3}+1)$ \\ 
$2^{m-1}+2^{(m-1)/2}$           & $(2^m-1)\cdot 2^{(m-3)/2}\cdot (2^{(m-1)/2}-1)\cdot (5\cdot 2^{m-1}+4)/3$ \\ 
$2^{m-1}+2^{(m+1)/2}$           & $(2^m-1)\cdot 2^{(m-5)/2}\cdot (2^{(m-3)/2}-1)\cdot (2^{m-1}-1)/3$ \\ \hline
\end{tabular}
\end{table}

\begin{theorem}\label{thm-BCHcodeDual}
Let $m \geq 5$ be an odd integer and let $\C_m$ be a binary code with the weight distribution of Table \ref{tab-zhou3}. 
Then the dual code $\C_{m}^\perp$ 
has parameters $[2^m-1, 2^m-1-3m, 7]$, and its weight distribution is given by  
\begin{eqnarray*}
2^{3m}A^\perp_k   &=&   \binom{2^m-1}{k}  +a  U_a(k) +  b  U_b(k)  + c  U_c(k) + d  U_d(k) + e  U_e(k), 
\end{eqnarray*}
where $0 \leq k \leq 2^m-1$, 
\begin{eqnarray*}
a &=&  (2^m-1) 2^{(m-5)/2} (2^{(m-3)/2}+1) (2^{m-1}-1)/3,  \\ 
b &=& (2^m-1) 2^{(m-3)/2} (2^{(m-1)/2}+1) (5\times 2^{m-1}+4)/3, \\ 
c &=& {(2^m-1)} (9\times 2^{2m-4}+3\times 2^{m-3}+1),  \\ 
d &=& (2^m-1) 2^{(m-3)/2} (2^{(m-1)/2}-1) (5\times 2^{m-1}+4)/3, \\ 
e &=&  (2^m-1) 2^{(m-5)/2} (2^{(m-3)/2}-1) (2^{m-1}-1)/3, 
\end{eqnarray*}
and 
\begin{eqnarray*}
U_a(k) =   \sum_{\substack{0 \le i \le 2^{m-1}-2^{(m+1)/2}  \\
0\le j \le 2^{m-1}+2^{(m+1)/2}-1 \\i+j=k}}(-1)^i  \binom{2^{m-1}-2^{(m+1)/2}} {i} \binom{2^{m-1}+2^{(m+1)/2}-1}{j}, 
\end{eqnarray*}
\begin{eqnarray*}
U_b(k)  =  \sum_{\substack{0 \le i \le 2^{m-1}-2^{(m-1)/2} \\
0\le j \le 2^{m-1}+2^{(m-1)/2}-1 \\i+j=k}}(-1)^i  \binom{2^{m-1}-2^{(m-1)/2}} {i} \binom{2^{m-1}+2^{(m-1)/2}-1}{j}, 
\end{eqnarray*}
\begin{eqnarray*}
U_c(k)  =  \sum_{\substack{0 \le i \le 2^{m-1} \\
0\le j \le 2^{m-1}-1 \\i+j=k}}(-1)^i \binom{2^{m-1}} {i}\binom{2^{m-1}-1} {j},   
\end{eqnarray*}
\begin{eqnarray*}
U_d(k)  =  \sum_{\substack{0 \le i \le 2^{m-1}+2^{(m-1)/2} \\ 
0\le j \le 2^{m-1}-2^{(m-1)/2}-1 \\i+j=k}}(-1)^i \binom{2^{m-1}+2^{(m-1)/2}}{i}\binom{2^{m-1}-2^{(m-1)/2}-1}{j},  
\end{eqnarray*}
\begin{eqnarray*} 
U_e(k)  &=&  \sum_{\substack{0 \le i \le 2^{m-1}+2^{(m+1)/2} \\
0\le j \le 2^{m-1}-2^{(m+1)/2}-1 \\i+j=k}}(-1)^i \binom{2^{m-1}+2^{(m+1)/2}}{i}\binom{2^{m-1}-2^{(m+1)/2}-1}{j}. 
\end{eqnarray*}
\end{theorem}

\begin{proof}
By assumption, the weight enumerator of $\C_{m}$ 
is given by 
$$
A(z)=1+az^{2^{m-1}-2^{(m+1)/2}} +bz^{2^{m-1}-2^{(m-1)/2}} +cz^{2^{m-1}}+dz^{2^{m-1}+2^{(m-1)/2}} +ez^{2^{m-1}+2^{(m+1)/2}}.
$$
It then follows from Theorem \ref{thm-MI} that the weight enumerator of $\C_{m}^\perp$ is given by
\begin{eqnarray*}
2^{3m}A^\perp(z) &=&  (1+z)^{2^m-1}\left[ 1 +
             a\left(\frac {1-z} {1+z}\right)^{2^{m-1}-2^{\frac{m+1}{2}}}  +  b\left(\frac {1-z} {1+z}\right)^{2^{m-1}-2^{ \frac{m-1}{2} }} \right] + \\
& &  (1+z)^{2^m-1}\left[
             c\left(\frac {1-z} {1+z}\right)^{2^{m-1}}+ d\left(\frac {1-z} {1+z}\right)^{2^{m-1}+2^{\frac{m-1}{2} }}  
             + e\left(\frac {1-z} {1+z}\right)^{2^{m-1}+2^{ \frac{m+1}{2} }} \right]. 
\end{eqnarray*}
Hence, we have 
\begin{eqnarray*}
2^{3m}A^\perp(z)      
&=&   (1+z)^{2^m-1} + \\ 
& &             a(1-z)^{2^{m-1}-2^{(m+1)/2}}(1+z)^{2^{m-1}+2^{(m+1)/2}-1}  + \\             
& &       b(1-z)^{2^{m-1}-2^{(m-1)/2}}(1+z)^{2^{m-1}+2^{(m-1)/2}-1}  + \\         
& &     c(1-z)^{2^{m-1}}(1+z)^{2^{m-1}-1} + \\
& &   d(1-z)^{2^{m-1}+2^{(m-1)/2}}(1+z)^{2^{m-1}-2^{(m-1)/2}-1}  + \\
& &   e(1-z)^{2^{m-1}+2^{(m+1)/2}}(1+z)^{2^{m-1}-2^{(m+1)/2}-1}.  
\end{eqnarray*}

Obviously, we have   
\begin{eqnarray*}
(1+z)^{2^m-1} &=& \sum_{k=0}^{2^m-1} \binom{2^m-1}{k}z^k. 
\end{eqnarray*}  
It is easily seen that 
\begin{eqnarray*}
(1-z)^{2^{m-1}-2^{(m+1)/2}}(1+z)^{2^{m-1}+2^{(m+1)/2}-1}  
=  \sum_{k=0}^{2^m-1} U_a(k) z^k 
\end{eqnarray*} 
and 
\begin{eqnarray*}
 (1-z)^{2^{m-1}-2^{(m-1)/2}}(1+z)^{2^{m-1}+2^{(m-1)/2}-1}   
=   \sum_{k=0}^{2^m-1} U_b(k) z^k. 
\end{eqnarray*} 
Similarly, 
\begin{eqnarray*}
(1-z)^{2^{m-1}+2^{(m-1)/2}}(1+z)^{2^{m-1}-2^{(m-1)/2}-1}  
=  \sum_{k=0}^{2^m-1} U_d(k) z^k 
\end{eqnarray*}
and 
\begin{eqnarray*} 
(1-z)^{2^{m-1}+2^{(m+1)/2}}(1+z)^{2^{m-1}-2^{(m+1)/2}-1} 
=  \sum_{k=0}^{2^m-1}  U_e(k) z^k. 
\end{eqnarray*} 
Finally, we have 
\begin{eqnarray*}
(1-z)^{2^{m-1}}(1+z)^{2^{m-1}-1}=  \sum_{k=0}^{2^m-1}  U_c(k) z^k.  
\end{eqnarray*}
Combining these formulas above yields the weight distribution formula for $A_k^\perp$. 

The weight distribution in Table \ref{tab-zhou3} tells us that the dimension of   $\C_{m}$ 
is $3m$. Therefore, the dimension of  $\C_{m}^\perp$ is equal to $2^m-1-3m$. Finally, we 
prove that the minimum distance of $\C_{m}^\perp$  equals $7$. 

We now prove that  $A_k^\perp=0$ for all $k$ with $1 \leq k \leq 6$. Let $x=2^{(m-1)/2}$. With the weight distribution 
formula obtained before, we have 
\begin{eqnarray*}
 \binom{2^m-1}{1} &=& 2x^2 - 1,\\
 a  U_a(1) &=& 1/3 x^7 + 7/12 x^6 - 2/3 x^5 - 7/8 x^4 + 5/12 x^3 + 7/24 x^2 - 1/12 x, \\
 b  U_b(1)  &=& 10/3x^7 + 5/3x^6 - 2/3x^5 + 1/2x^4 - 11/6x^3 - 2/3x^2 + 2/3x, \\
 c  U_c(1)  &=& -9/2x^6 + 3/4x^4 - 5/4x^2 + 1, \\
 d  U_d(1) &=& -10/3x^7 + 5/3x^6 + 2/3x^5 + 1/2x^4 + 11/6x^3 - 2/3x^2 - 2/3x, \\
e  U_e(1) &=& -1/3x^7 + 7/12x^6 + 2/3x^5 - 7/8x^4 - 5/12x^3 + 7/24x^2 + 1/12x.   
\end{eqnarray*} 
Consequently, 
\begin{eqnarray*}
2^{3m} A_1^\perp = \binom{2^m-1}{1}  +a  U_a(1) +  b  U_b(1)  + c  U_c(1) + d  U_d(1) + e  U_e(1)=0.   
\end{eqnarray*}

Plugging $k=2$ into the weight distribution formula above, we get that 
\begin{eqnarray*}
 \binom{2^m-1}{2} &=&  2x^4 - 3x^2 + 1,\\
 a  U_a(2) &=& 7/12x^8 + 5/6x^7 - 35/24x^6 - 13/12x^5 + 7/6x^4 + 1/6x^3 - 7/24x^2 + 1/12x,\\
 b  U_b(2)  &=& 5/3x^8 - 5/3x^7 - 7/6x^6 + 7/6x^5 - 7/6x^4 + 7/6x^3 + 2/3x^2 - 2/3x, \\
 c  U_c(2)  &=& -9/2x^8 + 21/4x^6 - 2x^4 + 9/4x^2 - 1, \\
 d  U_d(2) &=& 5/3x^8 + 5/3x^7 - 7/6x^6 - 7/6x^5 - 7/6x^4 - 7/6x^3 + 2/3x^2 + 2/3x, \\
e  U_e(2) &=& 7/12x^8 - 5/6x^7 - 35/24x^6 + 13/12x^5 + 7/6x^4 - 1/6x^3 - 7/24x^2 - 
    1/12x.   
\end{eqnarray*} 
Consequently, 
\begin{eqnarray*}
2^{3m} A_2^\perp = \binom{2^m-1}{2}  +a  U_a(2) +  b  U_b(2)  + c  U_c(2) + d  U_d(2) + e  U_e(2)=0.   
\end{eqnarray*}

After similar computations with the weight distribution formula, one can prove that  $A_k^\perp=0$ for all $k$ with 
$3 \leq k \leq 6$.

Plugging $k=7$ into the weight distribution formula above, we arrive at 
\begin{eqnarray*}
A_7^\perp =   \frac{(x^2-1) (2x^2 - 1) (x^4 - 5x^2 + 34)}{630} . 
\end{eqnarray*}
Notice that $x^4 - 5x^2 + 34=(x^2-5/2)^2 +34-25/4 >0$. We have $A_7^\perp >0$ for all odd $m \geq 5$. 
 This proves the desired conclusion on the minimum distance 
of $\C_{m}^\perp$.  
\end{proof}

\begin{theorem}\label{thm-lastcode} 
Let $m \geq 5$ be an odd integer and let $\C_m$ be a binary code with the weight distribution of Table \ref{tab-zhou3}. 
The code  $\overline{\C_{m}^\perp}^\perp$ has parameters 
$$ 
\left[2^m, \, 3m+1, \,  2^{m-1}-2^{(m+1)/2}\right],  
$$
and its weight enumerator is given by 
\begin{eqnarray}\label{eqn-doubledual}
\overline{A^\perp}^\perp (z) = 1+uz^{2^{m-1}-2^{\frac{m+1}{2} }} + vz^{2^{m-1}-2^{ \frac{m-1}{2}}} + wz^{2^{m-1}} + vz^{2^{m-1}+2^{ \frac{m-1}{2}}} + 
uz^{2^{m-1}+2^{\frac{m+1}{2} }} +z^{2^m}, 
\end{eqnarray} 
where 
\begin{eqnarray*}
u &=& \frac{2^{3m-4} - 3 \times  2^{2m-4} + 2^{m-3}}{3}, \\
v &=& \frac{5\times 2^{3m-2} + 3 \times 2^{2m-2} - 2^{m+1}}{3} , \\
w &=&  {2(2^m-1)} (9\times 2^{2m-4}+3\times 2^{m-3}+1). 
\end{eqnarray*}
\end{theorem}

\begin{proof}
It follows from Theorem \ref{thm-allcodes} that the code has all the weights given in (\ref{eqn-doubledual}). It remains to determine the frequencies 
of these weights. The weight distribution of the code $\C_{m}$ given in Table \ref{tab-zhou3}  and the generator matrix of the 
code $\overline{\C_{m}^\perp}^\perp$ documented in the proof of Theorem \ref{thm-allcodes} show that 
$$
\overline{A^\perp}^\perp_{2^{m-1}}=2c=w, 
$$
where $c$ was defined in Theorem \ref{thm-BCHcodeDual}. 

We now determine $u$ and $v$. Recall that $\C_{m}^\perp$ has minimum distance $7$. It then follows from Theorem 
\ref{thm-allcodes} that $\overline{\C_{m}^\perp}$ has minimum distance $8$.  The first and third Pless power moments 
say that 
\begin{eqnarray*} 
\left\{ \begin{array}{lll}
\sum_{i=0}^{2^m}  \overline{A^\perp}^\perp_{i}  &=& 2^{3m+1}, \\
\sum_{i=0}^{2^m} i^2 \overline{A^\perp}^\perp_{i}  &=& 2^{3m-1} 2^m(2^m+1).  
\end{array} 
\right. 
\end{eqnarray*} 
These two equations become 
\begin{eqnarray*} 
\left\{ \begin{array}{l}
 1+u+v+c = 2^{3m}, \\
 (2^{2m-2} +2^{m+1}) u +  (2^{2m-2} +2^{m-1}) v + 2^{2m-2}c + 2^{2m-1} = 2^{4m-2}(2^m+1).   
\end{array} 
\right. 
\end{eqnarray*} 
Solving this system of equations proves the desired conclusion on the weight enumerator of this code. 
\end{proof}

Finally, we settle the weight distribution of the code  $\overline{\C_{m}^\perp}$. 

\begin{theorem}\label{thm-3rdcode} 
Let $m \geq 5$ be an odd integer and let $\C_m$ be a binary code with the weight distribution of Table \ref{tab-zhou3}. 
The code  $\overline{\C_{m}^\perp}$ has parameters $[2^m, 2^m-1-3m, 8]$, and its weight 
distribution is given by 
\begin{eqnarray}
2^{3m+1}\overline{A^\perp}_k 
&=&  \left(1+(-1)^k \right) \binom{2^m}{k} + w E_0(k) + u E_1(k) +  v E_2(k) + v E_3(k) +u E_4(k), 
\end{eqnarray} 
where $w, u, v$ are defined in Theorem \ref{thm-lastcode}, and 
\begin{eqnarray*}
 E_0(k) = \frac{1+(-1)^k}{2} (-1)^{\lfloor k/2 \rfloor} \binom{2^{m-1}}{\lfloor k/2 \rfloor}, 
\end{eqnarray*} 
\begin{eqnarray*}
 E_1(k) = \sum_{\substack{0 \le i \le 2^{m-1}-2^{(m+1)/2} \\
0\le j \le 2^{m-1}+2^{(m+1)/2} \\i+j=k}}(-1)^i \binom{2^{m-1}-2^{(m+1)/2}} {i} \binom{2^{m-1}+2^{(m+1)/2}}{j},  
\end{eqnarray*}
\begin{eqnarray*}
 E_2(k) =  \sum_{\substack{0 \le i \le 2^{m-1}-2^{(m-1)/2} \\
0\le j \le 2^{m-1}+2^{(m-1)/2} \\i+j=k}}(-1)^i \binom{2^{m-1}-2^{(m-1)/2}} {i} \binom{2^{m-1}+2^{(m-1)/2}}{j},   
\end{eqnarray*} 
\begin{eqnarray*}
 E_3(k) =  \sum_{\substack{0 \le i \le 2^{m-1}+2^{(m-1)/2} \\
0\le j \le 2^{m-1}-2^{(m-1)/2} \\i+j=k}}(-1)^i \binom{2^{m-1}+2^{(m-1)/2}} {i} \binom{2^{m-1}-2^{(m-1)/2}}{j},  
\end{eqnarray*} 
\begin{eqnarray*}
 E_4(k) =  \sum_{\substack{0 \le i \le 2^{m-1}+2^{(m+1)/2} \\
0\le j \le 2^{m-1}-2^{(m+1)/2} \\i+j=k}}(-1)^i \binom{2^{m-1}+2^{(m+1)/2}} {i} \binom{2^{m-1}-2^{(m+1)/2}}{j},   
\end{eqnarray*} 
and $0 \leq k \leq 2^m$. 
\end{theorem} 

\begin{proof}
By definition, 
$$ 
\dim\left( \overline{\C_{m}^\perp} \right) = \dim\left(\C_{m}^\perp \right)=2^m-1-3m.   
$$
It has been showed in the proof of Theorem \ref{thm-BCHcodeDual} that the minimum distance of $\overline{\C_{m}^\perp}$ 
is equal to $8$. We now prove the conclusion on the weight distribution of this code. 

By Theorems \ref{thm-MI} and \ref{thm-lastcode}, the weight enumerator of  $\overline{\C_{m}^\perp}$ is given by
\begin{eqnarray}\label{eqn-july0}
2^{3m+1}\overline{A^\perp}(z) 
&=&  (1+z)^{2^m}\left[ 1 +   \left(\frac{1-z}{1+z}\right)^{2^m} +              
             w\left(\frac {1-z} {1+z}\right)^{2^{m-1}} \right] +  \nonumber \\
& &  (1+z)^{2^m}\left[ 
             u\left(\frac {1-z} {1+z}\right)^{2^{m-1}-2^{\frac{m+1}{2} }}+   v \left(\frac {1-z} {1+z}\right)^{2^{m-1}-2^{\frac{m-1}{2} }} 
            \right] +    \nonumber \\
      &&  (1+z)^{2^m}\left[ v \left(\frac {1-z} {1+z}\right)^{2^{m-1}+2^{\frac{m-1}{2} }}  + u\left(\frac {1-z} {1+z}\right)^{2^{m-1}+2^{\frac{m+1}{2}}}  \right].  
\end{eqnarray} 
Consequently, we have 
\begin{eqnarray}\label{eqn-j18-1}
2^{3m+1}\overline{A^\perp}(z) 
&=&  (1+z)^{2^m} + (1-z)^{2^m} + w (1-z^2)^{2^{m-1}} + \nonumber \\ 
 & &             u(1-z)^{2^{m-1}-2^{(m+1)/2}}(1+z)^{2^{m-1}+2^{(m+1)/2}} + \nonumber  \\
 & &  v(1-z)^{2^{m-1}-2^{(m-1)/2}}(1+z)^{2^{m-1}+2^{(m-1)/2}} + \nonumber \\ 
& &  v(1-z)^{2^{m-1}+2^{(m-1)/2}}(1+z)^{2^{m-1}-2^{(m-1)/2}} + \nonumber \\  
 & &             u(1-z)^{2^{m-1}+2^{(m+1)/2}}(1+z)^{2^{m-1}-2^{(m+1)/2}}.  
\end{eqnarray} 

We now treat the terms in (\ref{eqn-j18-1}) one by one. We first have 
\begin{eqnarray}\label{eqn-j18-2}
(1+z)^{2^m} + (1-z)^{2^m} = \sum_{k=0}^{2^m} \left(1+(-1)^k \right) \binom{2^m}{k}.  
\end{eqnarray} 
One can easily see that 
\begin{eqnarray}\label{eqn-j18-3}
(1-z^2)^{2^{m-1}} = \sum_{i=0}^{2^{m-1}} (-1)^i \binom{2^{m-1}}{i} z^{2i} = 
\sum_{k=0}^{2^{m}} \frac{1+(-1)^k}{2} (-1)^{\lfloor k/2 \rfloor} \binom{2^{m-1}}{\lfloor k/2 \rfloor} z^{k}.  
\end{eqnarray}  
Notice that 
\begin{eqnarray*}
(1-z)^{2^{m-1}-2^{(m+1)/2}}=\sum_{i=0}^{2^{m-1}-2^{(m+1)/2}} \binom{2^{m-1}-2^{(m+1)/2}}{i} (-1)^i z^i   
\end{eqnarray*} 
and 
\begin{eqnarray*}
(1+z)^{2^{m-1}+2^{(m+1)/2}}=\sum_{i=0}^{2^{m-1}+2^{(m+1)/2}} \binom{2^{m-1}+2^{(m+1)/2}}{i} z^i   
\end{eqnarray*} 
We have then 
\begin{eqnarray}\label{eqn-j18-4}
(1-z)^{2^{m-1}-2^{(m+1)/2}} (1+z)^{2^{m-1}+2^{(m+1)/2}} 
 = \sum_{k=0}^{2^m}  E_1(k) z^k. 
\end{eqnarray} 
Similarly, we have 
\begin{eqnarray}\label{eqn-j18-5}
 (1-z)^{2^{m-1}-2^{(m-1)/2}} (1+z)^{2^{m-1}+2^{(m-1)/2}} 
 = \sum_{k=0}^{2^m}  E_2(k) z^k,   
\end{eqnarray} 
\begin{eqnarray}\label{eqn-j18-6}
(1-z)^{2^{m-1}+2^{(m-1)/2}} (1+z)^{2^{m-1}-2^{(m-1)/2}} 
= \sum_{k=0}^{2^m}   E_3(k) z^k, 
\end{eqnarray}
\begin{eqnarray}\label{eqn-j18-7}
 (1-z)^{2^{m-1}+2^{(m+1)/2}} (1+z)^{2^{m-1}-2^{(m+1)/2}} 
 = \sum_{k=0}^{2^m}  E_4(k) z^k. 
\end{eqnarray}

Plugging (\ref{eqn-j18-2}), (\ref{eqn-j18-3}), (\ref{eqn-j18-4}), (\ref{eqn-j18-5}), (\ref{eqn-j18-6}), and (\ref{eqn-j18-7})  into 
(\ref{eqn-j18-1}) proves the desired conclusion.

\end{proof}

\section{Infinite families of $2$-designs from  $\C_{m}^\perp$ and  $\C_{m}$}\label{sec-2designs}

\begin{theorem}\label{thm-2designs} 
Let $m \geq 5$ be an odd integer and let $\C_m$ be a binary code with the weight distribution of Table \ref{tab-zhou3}. 
Let $\cP=\{0,1,2, \cdots, 2^m-2\}$, and let $\cB$ be the set of the supports of the codewords of $\C_{m}$
with weight $k$, where $A_k \neq 0$. Then $(\cP, \cB)$
is a $2$-$(2^m-1, k, \lambda)$ design, where
\begin{eqnarray*}
\lambda=\frac{k(k-1)A_k}{(2^m-1)(2^m-2)}, 
\end{eqnarray*} 
where $A_k$ is given in Table \ref{tab-zhou3}. 

Let $\cP=\{0,1,2, \cdots, 2^m-2\}$, and let $\cB^\perp$ be the set of the supports of the codewords of  $\C_{m}^\perp$ 
with weight $k$ and $A_k^\perp \neq 0$. Then $(\cP, \cB^\perp)$
is a $2$-$(2^m-1, k, \lambda)$ design, where
\begin{eqnarray*}
\lambda=\frac{k(k-1)A_k^\perp}{(2^m-1)(2^m-2)}, 
\end{eqnarray*}
where $A_k^\perp$ is given in Theorem \ref{thm-BCHcodeDual}.  
\end{theorem}

\begin{proof}
The weight distribution of $\C_{m}^\perp$ is given in Theorem \ref{thm-BCHcodeDual} and that of 
$\C_{m}$ is given in Table \ref{tab-zhou3}. By Theorem \ref{thm-BCHcodeDual},  the minimum 
distance $d^\perp$ of $\C_{m}^\perp$  is equal to $7$. Put $t=2$. The number of $i$ with $A_i \neq 0$ 
and $1 \leq i \leq 2^m-1 -t$ is $s=5$. Hence, $s=d^\perp-t$. The desired conclusions then follow from Theorem \ref{thm-AM1} 
and the fact that two binary vectors have the same support if and only if they are equal. 
\end{proof}

\begin{example} 
Let $m \geq 5$ be an odd integer and let $\C_m$ be a binary code with the weight distribution of Table \ref{tab-zhou3}. 
Then the BCH code $\C_{m}$ holds five $2$-designs with the following parameters: 
\begin{itemize} 
\item $(v,\, k, \, \lambda)=\left(2^m-1,\  2^{m-1}-2^{\frac{m+1}{2}}, \  
     \frac{2^{\frac{m-5}{2}} \left(2^{\frac{m-3}{2}}+1\right) \left(2^{m-1} - 2^{\frac{m+1}{2}} \right)\left(2^{m-1} - 2^{\frac{m+1}{2}} -1\right)}{6}  \right).$
\item $(v,\, k, \, \lambda)=\left(2^m-1,\  2^{m-1}-2^{\frac{m-1}{2}}, \  
     \frac{2^{m-2} \left(2^{m-1} - 2^{\frac{m-1}{2} } -1 \right)\left( 5 \times 2^{m-1} + 4 \right)}{6}  \right).$ 
\item $(v, \, k, \, \lambda)=\left(2^m-1, \ 2^{m-1}, \  2^{m-2} (9 \times 2^{2m-4}+ 3 \times 2^{m-3} +1) \right).$       
\item $(v,\, k, \, \lambda)=\left(2^m-1,\  2^{m-1}+2^{\frac{m-1}{2}}, \  
     \frac{2^{m-2} \left(2^{m-1} + 2^{\frac{m-1}{2} } -1 \right)\left( 5 \times 2^{m-1} + 4 \right)}{6}  \right).$     
\item $(v,\, k, \, \lambda)=\left(2^m-1,\  2^{m-1}+2^{\frac{m+1}{2}}, \  
     \frac{2^{\frac{m-5}{2}} \left(2^{\frac{m-3}{2}}-1\right) \left(2^{m-1} + 2^{\frac{m+1}{2}} \right)\left(2^{m-1} + 2^{\frac{m+1}{2}} -1\right)}{6}  \right).$ 
\end{itemize} 
\end{example}

\begin{example} 
Let $m \geq 5$ be an odd integer and let $\C_m$ be a binary code with the weight distribution of Table \ref{tab-zhou3}. 
Then the supports of all codewords of weight $7$ in  $\C_{m}^\perp$  give a $2$-$(2^m-1, 7, \lambda)$ 
design, where 
$$ 
\lambda=\frac{ 2^{2(m-1)} - 5 \times 2^{m-1} + 34 }{30}.  
$$
\end{example} 

\begin{proof}
By Theorem \ref{thm-BCHcodeDual}, we have 
$$ 
A^\perp_7=\frac{(2^{m-1}-1) (2^m-1) (2^{2(m-1)} - 5 \times 2^{m-1} + 34)}{630}. 
$$
The desired conclusion on $\lambda$ follows from Theorem \ref{thm-2designs}. 
\end{proof}

\begin{example} 
Let $m \geq 5$ be an odd integer and let $\C_m$ be a binary code with the weight distribution of Table \ref{tab-zhou3}. 
Then the supports of all codewords of weight $8$ in  $\C_{m}^\perp$  give a $2$-$(2^m-1, 8, \lambda)$ 
design, where 
$$ 
\lambda=\frac{ (2^{m-1}-4)(2^{2(m-1)} - 5 \times 2^{m-1} + 34) }{90}.  
$$
\end{example} 

\begin{proof}
By Theorem \ref{thm-BCHcodeDual}, we have 
$$ 
A^\perp_8=\frac{(2^{m-1}-1) (2^{m-1}-4) (2^m-1) (2^{2(m-1)} - 5 \times 2^{m-1} + 34)}{2520}. 
$$
The desired conclusion on $\lambda$ follows from Theorem \ref{thm-2designs}. 
\end{proof}

\begin{example} 
Let $m \geq 7$ be an odd integer and let $\C_m$ be a binary code with the weight distribution of Table \ref{tab-zhou3}. 
Then the supports of all codewords of weight $9$ in  $\C_{m}^\perp$  give a $2$-$(2^m-1, 9, \lambda)$ 
design, where 
$$ 
\lambda=\frac{ (2^{m-1}-4)(2^{m-1}-16)(2^{2(m-1)} -   2^{m-1} + 28) }{315}.  
$$
\end{example} 

\begin{proof}
By Theorem \ref{thm-BCHcodeDual}, we have 
$$ 
A^\perp_9=\frac{(2^{m-1}-1) (2^{m-1}-4) (2^{m-1}-16)  (2^m-1) (2^{2(m-1)} -   2^{m-1} + 28)}{11340}. 
$$
The desired conclusion on $\lambda$ follows from Theorem \ref{thm-2designs}. 
\end{proof}

\section{Infinite families of $3$-designs from  $\overline{\C_{m}^\perp}$ and  $\overline{\C_{m}^\perp}^\perp$}\label{sec-3designs}

\begin{theorem}\label{thm-newdesigns2}
Let $m \geq 5$ be an odd integer and let $\C_m$ be a binary code with the weight distribution of Table \ref{tab-zhou3}. 
Let $\cP=\{0,1,2, \cdots, 2^m-1\}$, and let $\overline{\cB^\perp}^\perp$ be the set of the supports of the codewords of $\overline{\C_{m}^\perp}^\perp$ 
with weight $k$, where $\overline{A^\perp}^\perp_k \neq 0$. Then $(\cP, \overline{\cB^\perp}^\perp)$
is a $3$-$(2^m, k, \lambda)$ design, where
\begin{eqnarray*}
\lambda=\frac{\overline{A^\perp}^\perp_k\binom{k}{3}}{\binom{2^m}{3}}, 
\end{eqnarray*}
where $\overline{A^\perp}^\perp_k$ is given in Theorem \ref{thm-lastcode}. 

Let $\cP=\{0,1,2, \cdots, 2^m-1\}$, and let $\overline{\cB^\perp}$ be the set of the supports of the codewords of
$\overline{\C_{m}^\perp}$ 
with weight $k$ and $\overline{A^\perp}_k \neq 0$. Then $(\cP, \overline{\cB^\perp})$
is a $3$-$(2^m, k, \lambda)$ design, where
\begin{eqnarray*}
\lambda=\frac{\overline{A^\perp}_k \binom{k}{3}}{\binom{2^m}{3}}, 
\end{eqnarray*}
where $\overline{A^\perp}_k$ is given in Theorem \ref{thm-3rdcode}.  
\end{theorem}

\begin{proof}
The weight distributions of $\overline{\C_{m}^\perp}^\perp$  and $\overline{\C_{m}^\perp}$   
are described in Theorems \ref{thm-lastcode} and \ref{thm-3rdcode}.   
Notice that the minimum distance $\overline{d^\perp}$ of $\overline{\C_{m}^\perp}^\perp$ is equal to $8$. Put $t=3$. The number
of $i$ with $\overline{A^\perp}_i \neq 0$ and $1 \leq i \leq 2^m -t$ is $s=5$. Hence, $s=\overline{d^\perp}-t$. Clearly, two binary vectors have the 
same support if and only if they are equal. 
The desired conclusions then follow from Theorem \ref{thm-AM1}. 
\end{proof} 

\begin{example} 
Let $m \geq 5$ be an odd integer and let $\C_m$ be a binary code with the weight distribution of Table \ref{tab-zhou3}. 
Then $\overline{\C_{m}^\perp}^\perp$  holds five $3$-designs with the following parameters: 
\begin{itemize}
\item $(v,\, k, \, \lambda)=\left(2^m,\  2^{m-1}-2^{\frac{m+1}{2}}, \  
    \frac{  \left(2^{m-1} - 2^{\frac{m+1}{2}} \right) \left(2^{m-1} - 2^{\frac{m+1}{2}} -1\right) \left(2^{m-1} - 2^{\frac{m+1}{2}} -2\right)}{48}  \right).$
\item $(v,\, k, \, \lambda)=\left(2^m,\  2^{m-1}-2^{\frac{m-1}{2}}, \  
     \frac{  2^{\frac{m-1}{2} } \left(2^{m-1} - 2^{\frac{m-1}{2} } -1 \right)  \left( 2^{\frac{m-1}{2} } -2 \right)\left( 5 \times 2^{m-3} + 1 \right)}{3}   \right).$

\item $(v, \, k, \, \lambda)=\left(2^m, \ 2^{m-1}, \  (2^{m-2}-1) (9 \times 2^{2m-4}+ 3 \times 2^{m-3} +1) \right).$      

\item $(v,\, k, \, \lambda)=\left(2^m,\  2^{m-1}+2^{\frac{m-1}{2}}, \  
     \frac{  2^{\frac{m-1}{2} } \left(2^{m-1} + 2^{\frac{m-1}{2} } -1 \right)  \left( 2^{\frac{m-1}{2} } +2 \right)\left( 5 \times 2^{m-3} + 1 \right)}{3}   \right).$

\item $(v,\, k, \, \lambda)=\left(2^m,\  2^{m-1}+2^{\frac{m+1}{2}}, \  
    \frac{  \left(2^{m-1} + 2^{\frac{m+1}{2}} \right) \left(2^{m-1} + 2^{\frac{m+1}{2}} -1\right) \left(2^{m-1} + 2^{\frac{m+1}{2}} -2\right)}{48}  \right).$
\end{itemize} 
\end{example}

\begin{example} 
Let $m \geq 5$ be an odd integer and let $\C_m$ be a binary code with the weight distribution of Table \ref{tab-zhou3}. 
Then the supports of all codewords of weight $8$ in  $\overline{\C_{m}^\perp}$ give a $3$-$(2^m, 8, \lambda)$ 
design, where 
$$ 
\lambda=\frac{2^{2(m-1)} - 5 \times 2^{m-1}+34}{30}. 
$$
\end{example} 

\begin{proof}
By Theorem \ref{thm-3rdcode}, we have 
$$ 
\overline{A^\perp}_8=\frac{2^m(2^{m-1}-1)(2^m-1)(2^{2(m-1)} - 5 \times 2^{m-1}+34)}{315}. 
$$
The desired value of $\lambda$ follows from Theorem \ref{thm-newdesigns2}.  
\end{proof}

\begin{example} 
Let $m \geq 7$ be an odd integer and let $\C_m$ be a binary code with the weight distribution of Table \ref{tab-zhou3}. 
Then the supports of all codewords of weight $10$ in  $\overline{\C_{m}^\perp}$ give a $3$-$(2^m, 10, \lambda)$ 
design, where 
$$ 
\lambda=\frac{(2^{m-1}-4) (2^{m-1}-16) (2^{2(m-1)}-2^{m-1}+28)}{315}. 
$$
\end{example} 

\begin{proof}
By Theorem \ref{thm-3rdcode}, we have 
$$ 
\overline{A^\perp}_{10}=\frac{2^{m-1}(2^{m-1}-1)(2^m-1)(2^{m-1}-4) (2^{m-1}-16) (2^{2(m-1)}-2^{m-1}+28) }{4 \times 14175}. 
$$
The desired value of $\lambda$ follows from Theorem \ref{thm-newdesigns2}.  
\end{proof}

\begin{example} 
Let $m \geq 5$ be an odd integer and let $\C_m$ be a binary code with the weight distribution of Table \ref{tab-zhou3}. 
Then the supports of all codewords of weight $12$ in  $\overline{\C_{m}^\perp}$ give a $3$-$(2^m, 12, \lambda)$ 
design, where 
$$ 
\lambda=\frac{(2^{h-2}-1) (2 \times 2^{5h}  - 55  \times 2^{4h} + 647  \times 2^{3h} - 2727  \times 2^{2h} + 11541  \times 2^{h} - 47208)}{2835}  
$$
and $h=m-1$.
\end{example} 

\begin{proof}
By Theorem \ref{thm-3rdcode}, we have 
$$ 
\overline{A^\perp}_{12}=\frac{\epsilon^2 (\epsilon^2-1) (\epsilon^2-4) (2 \epsilon^2-1) (2 \epsilon^{10} - 55 \epsilon^8 + 647 \epsilon^6 - 2727 \epsilon^4 + 11541 \epsilon^2 - 47208) }{8 \times 467775}, 
$$
where $\epsilon=2^{(m-1)/2}$. 
The desired value of $\lambda$ follows from Theorem \ref{thm-newdesigns2}.  
\end{proof}

\section{Two families of binary cyclic codes with the weight distribution of Table \ref{tab-zhou3}}\label{sec-examplecodes} 

To prove the existence of the $2$-designs in Section \ref{sec-2designs} and the $3$-designs in Section \ref{sec-3designs}, we present  
two families of binary codes of length $2^m-1$ with the weight distribution of Table \ref{tab-zhou3}.

Let $n=q^m-1$, where $m$ is a positive integer. Let $\alpha$ be a generator of $\gf(q^m)^*$. 
For any $i$ with $0 \leq i \leq n-1$, let $\m_i(x)$ denote the minimal polynomial of $\beta^i$ 
over $\gf(q)$. For any $2 \leq \delta \leq n$, define 
\begin{eqnarray}\label{eqn-BCHgeneratorPolyn}
g_{(q,n,\delta,b)}(x)=\lcm(\m_{b}(x), \m_{b+1}(x), \cdots, \m_{b+\delta-2}(x)), 
\end{eqnarray} 
where $b$ is an integer, $\lcm$ denotes the least common multiple of these minimal polynomials, and the addition 
in the subscript $b+i$ of $\m_{b+i}(x)$ always means the integer addition modulo $n$. 
Let $\C_{(q, n, \delta,b)}$ denote the cyclic code of length $n$ with generator 
polynomial $g_{(q, n,\delta, b)}(x)$. $\C_{(q, n, \delta, b)}$ is called  a \emph{primitive BCH code}  with \emph{designed distance} $\delta$. 
When $b=1$, the set $\C_{(q, n, \delta, b)}$ is called a \emph{narrow-sense primitive BCH code}.

Although primitive BCH codes are not good asymptotically, they are among the best linear codes when the length 
of the codes is not very large \cite[Appendix A]{DingBook}. So far, we have very limited knowledge of 
BCH codes, as the dimension and minimum distance of BCH codes are in general open, in spite of some 
recent progress \cite{Ding,DDZ15}. However, in a few cases the weight distribution of a BCH code can be 
settled. The following theorem introduces such a case. 

\begin{theorem}\label{thm-BCHcode}
Let $m \geq 5$ be an odd integer and let $\delta=2^{m-1}-1-2^{(m+1)/2}$. Then the BCH code $\C_{(2, 2^m-1, \delta, 0)}$ 
has length $n=2^m-1$, dimension $3m$, and the weight distribution in Table \ref{tab-zhou3}. 
\end{theorem}  

\begin{proof}
A proof can be found in \cite{DFZ}. 
\end{proof} 

It is known that the dual of a BCH code may not be a BCH code. The following theorem describes a family of cyclic codes 
having the weight distribution of Table \ref{tab-zhou3}, which may not be BCH codes. 

\begin{theorem}
Let $m \geq 5$ be an odd integer. Let $\C_m$ be the dual of the narrow-sense primitive  BCH code $\C_{(2, 2^m-1, 7, 1)}$. 
Then $\C_m$ has the weight distribution of Table \ref{tab-zhou3}.  
\end{theorem}

\begin{proof}
A proof can be found in \cite{Kasa69}. 
\end{proof}

\section{Summary and concluding remarks} 

In this paper, with any binary linear code of length $2^m-1$ and the weight distribution of Table \ref{tab-zhou3}, a huge number of infinite 
families of $2$-designs and $3$-designs with various block sizes were constructed. These constructions clearly show that the 
coding theory approach to constructing $t$-designs are in fact promising, and may stimulate further investigations in this direction.  
It is open if the codewords of a fixed weight in a family of linear codes can hold an infinite family of $t$-designs for some $t \geq 4$.  

It is noticed that the technical details of this paper are tedious. However, one has to settle the weight distribution of a linear code and the minimum distance of its dual at the same time, if one 
would like to employ the Assmus-Mattson Theorem for the construction of $t$-designs. Note that it 
could be very difficult to prove that a linear code has minimum weight $7$.   
This explains why the proofs of some of the theorems are messy and tedious, but necessary.

\section*{Acknowledgments} 

The research of C. Ding was supported by the Hong Kong Research Grants Council, under Project No. 16300415.

\end{document}